\documentclass[lettersize,journal]{IEEEtran}
\usepackage{amsmath,amsfonts}
\usepackage{algorithm}
\usepackage{array}
\usepackage[caption=false,font=normalsize,labelfont=sf,textfont=sf]{subfig}
\usepackage{textcomp}
\usepackage{stfloats}
\usepackage{calc}
\usepackage{amssymb}
\usepackage{url}
\usepackage{verbatim}
\usepackage{graphicx}
\usepackage{import}
\usepackage{array}
\usepackage{booktabs}
\usepackage{cite}
\usepackage{algpseudocode}
\usepackage[utf8]{inputenc}
\usepackage[dvipsnames]{xcolor}
\usepackage{siunitx}
\usepackage{pgfplots}
\pgfplotsset{compat=1.9}
\usepgfplotslibrary{groupplots}
\usepgfplotslibrary{dateplot}
\algdef{SE}[DOWHILE]{Do}{doWhile}{\algorithmicdo}[1]{\algorithmicwhile\ #1}%
\usepackage{bm}
\usepackage{amsthm}
\newtheorem{theorem}{Theorem}
\newtheorem{lemma}[theorem]{Lemma}
\definecolor{goldenrod}{HTML}{FFDF42}
\newcommand{\redline}{\raisebox{1.5pt}{\tikz{\draw[-,red,line width = 1pt](0,0) -- (7mm,0);}}}
\newcommand{\blueline}{\raisebox{1.5pt}{\tikz{\draw [blue, thick,dash pattern={on 5pt off 1pt on 0.85pt off 1.5pt}] (0,0) -- (7mm,0);}}}
\newcommand{\goldline}{\raisebox{1.5pt}{\tikz{\draw[-,color=goldenrod!30!orange,dotted,line width = 1pt](0,0) -- (7mm,0);}}}
\newcommand{\greenline}{\raisebox{1.5pt}{\tikz{\draw[-,black!47!green,dashed,line width = 1pt](0,0) -- (7mm,0);}}}
\newcommand{\blackline}{\raisebox{1.5pt}{\tikz{\draw[-,black,dashed,line width = 1pt](0,0) -- (7mm,0);}}}
\begin{document}

\title{Recursive Dynamic State Estimation for Power Systems with an Incomplete Nonlinear DAE Model}

\author{Milo\v{s} Katani{\'c}, John Lygeros,~\IEEEmembership{Fellow,~IEEE,} Gabriela Hug,~\IEEEmembership{Senior Member,~IEEE}
        % <-this % stops a space
\thanks{Research supported by the Swiss National Science Foundation under NCCR Automation, grant agreement 51NF40\_180545}% <-this % stops a space
\thanks{Milos Katanic, John Lygeros, and Gabriela Hug are with the Information Technology and Electrical Engineering (D-ITET), ETH Zurich, 8092 Zurich, Switzerland (e-mail: mkatanic@ethz.ch; jlygeros@ethz.ch; hug@eeh.ee.ethz.ch).}
}

%\markboth{IEEE TRANSACTIONS ON POWER SYSTEMS}%
%{Shell \MakeLowercase{\textit{et al.}}: A Sample Article Using IEEEtran.cls for %IEEE Journals}


\maketitle

\begin{abstract}
Power systems are highly complex, large-scale engineering systems subject to many uncertainties, which makes accurate mathematical modeling challenging. This paper proposes a novel, centralized dynamic state estimator for power systems that lack models of some components. Including the available dynamic evolution equations, algebraic network equations, and phasor measurements, we apply the least squares criterion to estimate all dynamic and algebraic states recursively. The approach results in an algorithm that generalizes the iterated extended Kalman filter and does not require static network observability. We further derive a graph theoretic condition for placing phasor measurement units that guarantees the uniqueness of the solution.  A numerical study evaluates the performance under short circuits in the network and load changes and shows superior tracking performance compared to robust procedures from the literature within computational times that are feasible for real-time application. 
\end{abstract}


\section{Introduction}
Static state estimation (SSE) is an important component of energy management systems in power system control centers \cite{abur}. Traditionally, SSE relies on real-time measurement data obtained by the Supervisory Control and Data Acquisition (SCADA) system to calculate the most probable system state. Thereby, the system dynamics are disregarded due to slow scan rates and the asynchronous nature of the SCADA measurements. \par
Recent years have witnessed an increase in the deployment of phasor measurement units (PMUs), devices that directly measure current and voltage phasors at high sampling rates and time-stamp those measurements against a time reference obtained by the Global Positioning System. The advent of PMUs promoted the research interest in dynamic state estimation (DSE) for power systems. Even though the number of PMUs is constantly increasing, the growth rate in the North American power grid may have slowed down in the last few years, according to \cite{PMU_report}. This fact motivates the development of state estimation methods that do not rely on static network observability by PMUs. 
\subsection{Related work}
DSE for power systems has a rich and growing literature. Initial works considered only the swing dynamics of generators and completely neglected voltage dynamics in the system \cite{scholtz_thesis}. Subsequently, the research community focused on the integration of dynamic models of generators and their controllers, thus still neglecting the estimation of the algebraic states (complex nodal voltages) \cite{pal}. Along this line, the most represented methods are the nonlinear variants of the Kalman filter: the extended Kalman filter \cite{ekf}, the unscented Kalman filter \cite{Qi_unsc}, and the particle filter \cite{particle}. Because of the inherent complexity of dealing with large-scale, nonlinear differential-algebraic equations (DAEs), researchers have leveraged \mbox{index-1} characteristics of power system models to first propagate the differential states and subsequently solve the algebraic equations to obtain the voltage estimates \cite{sakis}. Another simplification commonly employed is the Kron reduction \cite{c_dse}, \cite{c_dse2}, which models the loads as prespecified impedances. An alternative approach that assumes static network observability is to use two-stage estimators where, in the first stage, only algebraic states are calculated, and subsequently, generators' dynamic states are estimated \cite{abur_dynamic}. None of the referenced works aims to estimate differential and algebraic states simultaneously, a comprehensive approach that ensures the optimality of all state estimates. \par
The implementation of DSE directly on the nonlinear DAE models to jointly estimate both differential and algebraic states is relatively recent. Along this line, \cite{nadeem} proposes a continuous time $H_{\inf}$ observer to simultaneously track differential and algebraic states. An extension to this observer to include renewables generation is provided in \cite{taha1}. Reference \cite{limits} shows that inequality constraints can be explicitly incorporated into the least squares estimation by reformulating them into equality constraints using complementarity theory. However, these approaches assume the availability of the complete system model, including all generators, loads, and the grid. In our previous work \cite{Milos}, we present initial efforts to address the problem of partially known system models in a moving-horizon fashion, albeit with a substantial computational burden.\par
We argue that one of the main challenges of the centralized DSE is the availability and complexity of accurate mathematical models \cite{motivation}. Namely, power systems are inherently complex, heterogeneous systems of transcontinental scale built incrementally over many decades. Consequently, they are often modeled by nonlinear equations of large dimensions spanning multiple time scales. Therefore, obtaining accurate models of all components is a daunting task. In addition, different areas of a power system are often operated by independent utilities, which may be reluctant to share models and measurement data. This problem may be exacerbated by the penetration of renewable energy sources as they are usually connected to the grid via vendor-specific power electronic converters, whose behavior is, to a large extent, dependent on the employed control algorithms. Above all, various unanticipated disturbances and faults can occur, such as short circuits and line or load disconnections. It is unrealistic to assume that the operators can have real-time knowledge of all changes/disturbances occurring in the system.  This fact makes it likely for operators to work with only partially known power system models. \par
The aforementioned difficulty has been addressed in the literature using decentralized state estimation or simplified models with robust estimators to subdue the effect of innovation or observation outliers. Decentralized state estimation aims to estimate the states of each generator separately \cite{abhinav, Singh}. The limitation of this method is that it requires the placement of a PMU on each generator terminal. In addition, the voltages at nodes where no generator is connected cannot be estimated. The second approach leverages short-term injection forecasts, historical records, or similar approximation methods to model loads or renewables infeed together with robust estimation schemes \cite{zhao}, \cite{nadeem}. This approach is similar to adding pseudo measurements with high covariance noise to SSE to restore observability when the network is unobservable. To deal with model and measurement uncertainties and unknown generator inputs, various robust estimators have been proposed \cite{robust3, robust4, robust5, robust6, robust_new}. However, robust estimation methods are usually computationally more expensive and more susceptible to numerical instabilities than ordinary least squares, which may prove prohibitive for real-time application.\par
References \cite{neuro} and \cite{GOLEIJANI} address the problem of an incomplete power system model, but both works are based on artificial neural networks; hence, their out-of-sample performance is still to be validated, while \cite{Abooshahab2022} deals with a simplified, linear power system model. To the best of our knowledge, no method for optimal DSE for power systems with a partially known nonlinear DAE model exists at present. 
\subsection{Contributions}
This paper aims to fill the above-mentioned research gap. Specifically, the contributions can be stated as follows:
\begin{itemize}
    \item We propose a centralized, recursive, least squares estimator for algebraic and dynamic state estimation for power systems with a partially known model. The method allows one to isolate a part of the system and estimate differential and algebraic states within without making any assumptions about the rest of the system. To derive the recursion, we extend the Kalman filtering for under-determined linear DAE systems from \cite{ishihara} to nonlinear cases. This results in a generalization of the iterated extended Kalman filter (IEKF). Moreover, it provides the possibility of adding new or removing existing models. Thus, for example, once a dynamic model of a generator or a load becomes available, it can be integrated into the estimation framework without redesigning the estimator. And vice versa, if a model proves suspicious, it can be omitted as long as estimability is still guaranteed;
    \item This paper achieves unique state reconstruction and introduces the notion of \textit{topological estimability}, which provides a location-dependent minimum number of real-time measurements that must be placed in the system. Estimability provides a guarantee that all system states can be uniquely estimated, given the available models and measurements (see \cite{kalman_descriptor} for a rigorous definition). The proposed graph theoretic approach is purely topological, independent of system parameters or operating points. We show that estimability can be achieved with fewer measurements than static network observability. In addition, there is no requirement to place PMUs at any specific node type. 
\end{itemize}

\subsection{Outline}
The remainder of the paper is structured as follows. Section \ref{Estimation} proposes the estimation algorithm. Section \ref{Estimability} addresses the topological estimability condition. In Section \ref{Results}, simulation results are discussed. Finally, Section \ref{Conclusion} concludes the paper.

\section{Estimation}
\label{Estimation}

\subsection{Power System Modeling}
\label{subsec_ekf}
In this work, we consider a power system that lacks models of some components. We call these devices unknown injectors, and they can be static or dynamic components that inject positive or negative power into the grid. They can be generators or loads with unknown models or, for example, distribution grids whose models, parameters, and real-time measurements are inaccessible to transmission system operators. Due to unknown injectors, the employed power system estimation model is incomplete and is given by the following continuous-time, stochastic, under-constrained, nonlinear, semi-explicit DAE system:
\begin{subequations}
\label{eq:power_system}
\begin{equation}
    \dot{\bm{{{y}}}} = \Tilde{\bm{f}}(\bm{{{y}}}, \bm{v}) + \Tilde{\bm{\xi}}_d,
    \label{eq:differential}
    \end{equation}
    \begin{equation}
    \bm{0} = \Tilde{\bm{g}}(\bm{{{y}}}, \bm{v}) + \Tilde{\bm{\xi}}_a,
    \label{eq:algebraic}
    \end{equation}
\end{subequations}
where $\bm{{{y}}} \in \mathbb{R}^{n_d}$ and $\bm{v} \in \mathbb{R}^{n_a}$ are respectively the vectors of the differential and algebraic states and $\Tilde{\bm{f}}: \mathbb{R}^{n_d + n_a} \to \mathbb{R}^{n_d}$ and $\Tilde{\bm{g}}: \mathbb{R}^{n_d + n_a} \to \mathbb{R}^{{n}_g}$, with ${n}_g<n_a$, are nonlinear, differentiable, Lipschitz continuous functions representing the time evolution and algebraic nodal current balance equations in rectangular coordinates (see \cite{model_realimag} for details). $\Tilde{\bm{\xi}}_d \in \mathbb{R}^{n_d}$ and $\Tilde{\bm{\xi}}_a \in \mathbb{R}^{{n}_g}$ are continuous-time Gaussian white noise vectors of independent random variables with  $\Tilde{\bm{\xi}}_d \sim (\bm{0}, \Tilde{\bm{Q}}_d)$ and $\Tilde{\bm{\xi}}_a \sim (\bm{0}, \Tilde{\bm{Q}}_a)$, where $\Tilde{\bm{Q}}_d \in \mathbb{R}^{n_d \times n_d}$ and $\Tilde{\bm{Q}}_a \in \mathbb{R}^{{n}_g \times {n}_g}$ are positive definite noise covariances. Internal states of individual dynamic models are coupled through the network, which is assumed to be connected. Inputs (set points) are not explicitly stated in (\ref{eq:differential})--(\ref{eq:algebraic}), but are regarded as known parameters. As the number of states is larger than the number of equations,  this system of equations is under-determined, and hence, the initial value problem (IVP) is not uniquely solvable. Therefore, we rely on time-discrete measurements to uniquely infer the system trajectory. Note that as a consequence, the standard approaches to Kalman filtering for nonlinear DAE systems from the literature \cite{mandela} cannot be applied as they require the invertibility of the Jacobian $\frac{\partial  \Tilde{\bm{g}} }{\partial \bm{v}}$ in (\ref{eq:algebraic}). \par
One may be tempted to complete the power system model by modeling the unknown components with their expected injection obtained from historical values or other forecasting methods. However, in this paper, we show that wrong forecasts can cause the estimation to converge to a wrong system state even if robust estimation methods are used. Such hour or day-ahead forecasts are especially unreliable during faults or abnormal operating conditions, during which it is arguably of the utmost importance for DSE to deliver unbiased state estimates. \par
We use the following discrete-time representation:
\begin{subequations}
\begin{align}
    \bm{{{y}}}_{k } &= \bm{f}(\bm{{{y}}}_{k-1}, \bm{v}_{k-1}, \bm{{{y}}}_{k}, \bm{v}_{k})h + \bm{{{y}}}_{k-1} + \bm{\xi}_{d,k}, \label{eq:dae_discrete_diff}\\
    \bm{0} &= \bm{g}(\bm{{{y}}}_{k}, \bm{v}_{k}) + \bm{\xi}_{a,k},\label{eq:dae_discrete_alg}\\
    \bm{z}_{k} &= \begin{bmatrix} \bm{0} & \bm{C}_2 \end{bmatrix} \begin{bmatrix} \bm{{y}}_k\\ \bm{v}_{k}\end{bmatrix} + \bm{\nu}_{k} \label{eq:dae_discrete_meas},
\end{align}
\end{subequations}
where $\bm{f}:  \mathbb{R}^{2n_d + 2n_a} \to \mathbb{R}^{n_d}$ represents the discretized dynamics of (\ref{eq:differential}), ${\bm{g}}(\cdot)\equiv \Tilde{\bm{g}}(\cdot)$, $\bm{z}_{k} \in \mathbb{R}^{n_m}$ collects all PMU measurements, $\bm{C} := \begin{bmatrix} \bm{0} & \bm{C}_2 \end{bmatrix} \in \mathbb{R}^{n_m \times (n_d + n_a)}$ is the output matrix of the linear PMU measurement function in rectangular coordinates \cite{linear_PMU}, $h \in \mathbb{R}^+$ denotes the discretization step size, and the subscript $k \in \mathbb{Z}$ denotes the time instant. Algebraic equations are written for the following time step $k$. Discrete-time noise vectors $\bm{{\xi}}_{d,k} \in \mathbb{R}^{n_d}$, $\bm{{\xi}}_{a,k} \in \mathbb{R}^{n_g}$, and $\bm{{\nu}}_{k} \in \mathbb{R}^{{n}_m}$ are assumed to be independent Gaussian zero mean with
\begin{equation}
\mathbb{{E}}\left[
        \begin{bmatrix}
            \bm{{\xi}}_{d,k}\\ \bm{{\xi}}_{a ,k} \\\bm{{\nu}}_k
        \end{bmatrix}\begin{bmatrix}
             \bm{{\xi}}_{d,j}\\ \bm{{\xi}}_{a ,j} \\\bm{{\nu}}_j
        \end{bmatrix}^T \right]=\begin{bmatrix}
            \bm{{Q}}_{d} & \bm{0} & \bm{0} \\ \bm{0} & \bm{{Q}}_{a} & \bm{0}\\
            \bm{0} & \bm{0} & \bm{R}
        \end{bmatrix}\delta(k-j),
\end{equation}
where $\bm{{Q}}_{d} \in \mathbb{R}^{n_d \times n_d} \succ 0$, $\bm{{Q}}_{a} \in \mathbb{R}^{n_g \times n_g} \succ 0$, $\bm{{R}} \in \mathbb{R}^{n_m \times n_m} \succ 0$. Note that the process noise $\bm{{\xi}}_{d,k}$ may be non-zero mean as it also incorporates a discretization error, which is difficult to quantify. This error, in general, also depends on the voltage evolution between the measurement samples; however, for deriving the optimal recursive filter, we ignore these attributes. 
In (\ref{eq:dae_discrete_meas}), conventional measurement devices, such as remote terminal units, are not considered as their slow scan rate and the lack of synchronization hinder their employment during transient conditions \cite{motivation}. Formulation (\ref{eq:dae_discrete_diff}) subsumes several one-step discretization methods, such as the explicit Euler, the implicit Euler, and the trapezoidal rule. Note that it is not possible to directly accommodate multi-stage discretization schemes (such as Runge-Kutta methods). The reason is that the power system model is under-determined; therefore, online measurements, which arrive at pre-determined sampling points, are necessary to guarantee the uniqueness of the states' trajectories. However, no measurements are available at intermediary stages introduced by the discretization methods. On the other hand, employing multi-step methods would undermine the advantage of recursive estimation as it would require storing multiple previous state estimates in the computer memory.
\subsection{Dynamic State Estimation Algorithm}
To derive the optimal filter for nonlinear systems, we extend the approach from \cite{ishihara}, outlined in Appendix \ref{subsec_kf}, by first linearizing the dynamic equations and subsequently applying the optimal recursive filtering to the obtained affine system. Furthermore, we propose to apply the scheme several times iteratively, where each iteration yields a new (better) state estimate and, hence, a better linearization point. In the literature, such an iterative estimation procedure is called the iterated extended Kalman filter (IEKF) \cite{simon2006optimal}. 
The idea is that the successive linearization will render the linearization error negligible compared to noise. The approach is similar to applying the Gauss-Newton method to solve the original nonlinear optimization problem \cite{bell}. Denoting $\bm{\hat{{y}}}_{k|j}, k\leq j$, as the estimate of $\bm{{{y}}}_{k}$ once $\bm{{z}}_{j}$ has been processed, we assume that $(\bm{\hat{{{y}}}}_{k-1|k-1},  \bm{\hat{v}}_{k-1|k-1})$, and its covariance  $\bm{P}_{k-1|k-1}  \in \mathbb{R}^{(n_d + n_a) \times (n_d + n_a)} \succ 0$ are available.\par 
%$\mathbb{E}\left[\begin{bmatrix}\bm{\hat{{{y}}}}_{k-1|k-1} - \bm{{{y}}}_{k-1}\\ \bm{\hat{v}}_{k-1|k-1} - \bm{v}_{k-%1}\end{bmatrix}\begin{bmatrix}\bm{\hat{{{y}}}}_{k-1|k} - \bm{{{y}}}_{k-1}\\ \bm{\hat{v}}_{k-1|k-1} - \bm{v}_{k-1} \end{bmatrix}^T \right]=\bm{P}_{k-%1|k-1}$ are available, where $\bm{P}_{k-1|k-1}  \in \mathbb{R}^{(n_d + n_a) \times (n_d + n_a)} \succ 0$. 
First note that (\ref{eq:dae_discrete_diff})--(\ref{eq:dae_discrete_alg}) need to be linearized around the previous estimate $(\hat{\bm{{{y}}}}_{k-1|k-1}, \bm{\hat{v}}_{k-1|k-1})$, but also around the current iterate $i$ of the current estimate $\left(\hat{\bm{{{y}}}}_{k|k}^{(i)}, \bm{\hat{v}}_{k|k}^{(i)}\right)$. The initial iterate of the current time step can be obtained by setting it equal to the previous estimate, i.e., $\hat{\bm{{{y}}}}_{k|k}^{(0)}=\hat{\bm{{{y}}}}_{k-1|k-1}, \quad \bm{\hat{v}}_{k|k}^{(0)}=\bm{\hat{v}}_{k-1|k-1}.$
The Taylor series linearization of  (\ref{eq:dae_discrete_diff})--(\ref{eq:dae_discrete_alg}) around $\left({\bm{\hat{{{y}}}}_{k-1\vert k-1}, \bm{\hat{v}}_{k-1\vert k-1},\bm{\hat{{{y}}}}^{(i)}_{k \vert k}, \bm{\hat{v}}^{(i)}_{k\vert k}}\right)$ gives: 
\begin{align}
    \bm{{{y}}}_{k} \approx &~ h\bm{f}\left({\bm{\hat{{{y}}}}_{k-1\vert k-1}, \bm{\hat{v}}_{k-1\vert k-1},\bm{\hat{{{y}}}}^{(i)}_{k \vert k}, \bm{\hat{v}}^{(i)}_{k\vert k}}\right) \nonumber  \\
    &+ h\frac{\partial\bm{ f}}{\partial \bm{{{y}}}_{k-1}}  ( \bm{{{y}}}_{k-1} \!-\! \bm{\hat{{{y}}}}_{k-1\vert k-1} ) \nonumber\\
     &+h\frac{\partial\bm{ f}}{\partial \bm{v}_{k-1}} ( \bm{v}_{k-1} - \bm{\hat{v}}_{k-1\vert k-1} ) \nonumber \\
     &+ h\frac{\partial\bm{ f}}{\partial \bm{{{y}}}_{k}}  \left( \bm{{{y}}}_{k} - \bm{\hat{{{y}}}}^{(i)}_{k\vert k} \right)+
     h\frac{\partial\bm{ f}}{\partial \bm{v}_{k}}  \left( \bm{v}_{k}\! - \bm{\hat{v}}^{(i)}_{k\vert k} \right) \nonumber  \\
      &+ \bm{{{{y}}}}_{k-1}  +\bm{\xi}_{d,k},
      \label{eq:iekf_diff}\\
    \bm{0} \approx&~ \bm{g}\left(\bm{\hat{{{y}}}}^{(i)}_{k\vert k}, \bm{\hat{v}}^{(i)}_{k \vert k}\right) +
    \frac{\partial\bm{ g}}{\partial \bm{{{y}}}_{k}}\left( \bm{{{y}}}_{k} - \bm{\hat{{{y}}}}^{(i)}_{k\vert k} \right) \nonumber  \\&+ \frac{\partial\bm{ g}}{\partial \bm{v}_{k}}\left( \bm{v}_{k} - \bm{\hat{v}}^{(i)}_{k\vert k} \right)+\bm{\xi}_{a,k}. 
    \label{eq:iekf_alg}
\end{align}

In what follows, we omit the iteration superscript $i$ as the analysis focuses on one single iteration of the algorithm. Introducing auxiliary variables for the Jacobians
\begin{align}
    \bm{A}_{1,k-1} &:= \frac{\partial\bm{ f}}{\partial \bm{{{y}}}_{k-1}} \in \mathbb{R}^{n_d \times n_d},& \bm{A}_{2,k-1}&:= \frac{\partial\bm{ f}}{\partial \bm{v}_{k-1}}\in \mathbb{R}^{n_d \times n_a}, \nonumber \\
    \bm{E}_{1,k} &:= \frac{\partial\bm{ f}}{\partial \bm{{{y}}}_{k}}\in \mathbb{R}^{{n}_d \times n_d}, &\bm{E}_{2,k} &:= \frac{\partial\bm{ f}}{\partial \bm{v}_{k}}\in \mathbb{R}^{{n}_d \times n_a}, \nonumber \\
    \bm{E}_{3,k} &:= \frac{\partial\bm{g}}{\partial \bm{{{y}}}_{k}}\in \mathbb{R}^{{n}_g \times n_d},& \bm{E}_{4,k} &:= \frac{\partial\bm{g}}{\partial \bm{v}_{k}}\in \mathbb{R}^{{n}_g \times n_a}\label{jac_E2}, \nonumber
\end{align}
and concatenating 
\begin{align}\bm{x}_k &:= \begin{bmatrix} \bm{{{y}}}_k\\ \bm{v}_k  \end{bmatrix}  \in \mathbb{R}^{n_d + n_a}, \bm{\xi}_k := \begin{bmatrix}
    \bm{\xi}_{d,k} \nonumber \\ \bm{\bm{\xi}}_{a,k}
\end{bmatrix}  \in \mathbb{R}^{n_d+n_g},\\ \bm{Q} &:= \begin{bmatrix} \bm{Q}_{d} & \bm{0} \\ \bm{0} & \bm{Q}_{a} \end {bmatrix}  \in \mathbb{R}^{(n_d + {n}_g) \times (n_d + {n}_g)},
\end{align}
allows us to rewrite (\ref{eq:iekf_diff})--(\ref{eq:iekf_alg}) and (\ref{eq:dae_discrete_meas}) as
\begin{align}
    \bm{E}_{k}\bm{x}_{k} &\approx \bm{A}_{k-1}\bm{x}_{k-1} + \bm{\Delta}_k + \bm{\xi}_k, \\
    \bm{z}_{k} &= \bm{C} \bm{x}_{k} + \bm{\nu}_{k}, 
\end{align}
where $\bm{E}_{k}, \bm{A}_{k-1},$ and $\bm{\Delta}_k$ are defined in (\ref{eq:long}). Note that because of the choice of the discrete-time representation (\ref{eq:dae_discrete_diff})-(\ref{eq:dae_discrete_alg}), $\bm{E}_{3,k}$ and $\bm{E}_{4,k}$ are not zero matrices, which would be the case in the continuous time in (\ref{eq:power_system}). \par
\begin{figure*}[bp]
%\line(1,0){\textwidth}
\rule{\textwidth}{0.5pt}
\begin{align}
\bm{E}_{k} &:= \begin{bmatrix}h\bm{E}_{1,k} - \bm{I} & h\bm{E}_{2,k} \\ \bm{E}_{3,k} & \bm{E}_{4,k} \end{bmatrix} \in \mathbb{R}^{(n_d+{n}_g) \times (n_d+{n}_a)},
\bm{A}_{k-1} := \begin{bmatrix} -h\bm{A}_{1,k-1} - \bm{I} & -h\bm{A}_{2,k-1}\\ \bm{0} & \bm{0} \end{bmatrix} \in \mathbb{R}^{(n_d+{n}_g) \times (n_d+{n}_a)}, \nonumber \\
       \bm{\Delta}_k &:= \begin{bmatrix} h\bm{E}_{1,k}\bm{\hat{y}}_{k\vert k} + h\bm{E}_{2,k}\bm{\hat{v}}_{k \vert k} +h\bm{A}_{1,k-1}\bm{\hat{y}}_{k-1\vert k-1} + h\bm{A}_{2,k-1}\bm{\hat{v}}_{k-1 \vert k-1}- h\bm{f}({\bm{\hat{{{y}}}}_{k-1\vert k-1}, \bm{\hat{v}}_{k-1\vert k-1},\bm{\hat{{{y}}}}^{(i)}_{k \vert k}, \bm{\hat{v}}^{(i)}_{k\vert k}}) \\   \bm{E}_{3, k} \bm{\hat{y}}_{k \vert k} + \bm{E}_{4,k} \bm{\hat{v}}_{k \vert k}  -\bm{g}(\bm{\hat{{{y}}}}^{(i)}_{k \vert k}, \bm{\hat{v}}^{(i)}_{k\vert k}) \end{bmatrix} \in \mathbb{R}^{n_d+{n}_g}. \label{eq:long}      
\end{align}
\end{figure*}
We can now formulate the optimization problem to calculate the least-squares estimate at time step $k$, in a similar way as for linear systems (see (\ref{opt_pdf}) in Appendix \ref{subsec_kf})
\begin{equation}
    \min_{\mathfrak{X}} \quad (\mathfrak{A}\mathfrak{X}-\mathfrak{B})^T\mathfrak{R}(\mathfrak{A}\mathfrak{X}-\mathfrak{B}), \label{opt_ekf}
\end{equation}
where
\begin{align}
\mathfrak{A}&:=\begin{bmatrix}\bm{E}_k & \bm{A}_{k-1}\\ \bm{C} &\bm{0} \\ \bm{0}& \bm{I} \end{bmatrix}, \mathfrak{B}:=\begin{bmatrix}\bm{\Delta}_k\\ \bm{z}_k \\ - \bm{\hat{x}}_{k-1|k-1} \end{bmatrix}, \nonumber \\
\mathfrak{R} &:= \begin{bmatrix} \bm{Q}^{-1} & \bm{0} & \bm{0} \\ \bm{0} & \bm{R}^{-1} & \bm{0} \\ \bm{0} & \bm{0} & \bm{P}^{-1}_{k-1|k-1} \end{bmatrix}, \mathfrak{X} : = \begin{bmatrix} \bm{x}_k \\ -\bm{x}_{k-1}\end{bmatrix}.
\end{align}
Compared to the linear case in (\ref{opt_pdf}), the additional term $\bm{\Delta}_k$ is introduced by the linearization. Under the assumption $\left[\begin{smallmatrix}\bm{E}_k\\\bm{C} \end{smallmatrix} \right]$ has full column rank, the minimum is attained at
\begin{equation}
\hat{\mathfrak{X}}=(\mathfrak{A}^T\mathfrak{R}\mathfrak{A})^{-1}\mathfrak{A}^T\mathfrak{R} \mathfrak{B},
\end{equation}
or
\begin{equation}
\begin{aligned}
    \begin{bmatrix} \bm{\hat{x}}_{k|k} \\ \bm{\hat{x}}_{k-1|k} \end{bmatrix}= &\begin{bmatrix} \bm{P}_{k|k} & \bm{P}_{k,k-1|k}\\ \bm{P}_{k,k-1|k}^T & \bm{P}_{k-1|k} \end{bmatrix}\cdot \\
    &\begin{bmatrix}\bm{C}^T \bm{R}^{-1}\bm{z}_k + \bm{E}_k^T\bm{Q}^{-1}\bm{\Delta}_k \\ -\bm{P}_{k-1|k-1}^{-1} \bm{\hat{x}}_{k-1|k-1} + \bm{A}_{k-1}^T\bm{Q}^{-1}\bm{\Delta}_k\end{bmatrix}.
    \label{kalman_big_matrix}
\end{aligned}
\end{equation}
To simplify the notation, we drop the subscript after the vertical bar, as it is implicitly clear from the time step at which an estimate was obtained. As we are only interested in the estimation of the most recent state of the system $\bm{\hat{x}}_{k}$, rewriting the terms from the linear case (\ref{state_kalman}) and including the additional terms introduced by the linearization gives
\begin{align}
    \bm{\hat{x}}_{k}\overset{(1)}{=} &~\bm{P}_{k}\bm{E}^T(\bm{Q} + \bm{A}\bm{P}_{k-1}\bm{A}^T)^{-1}\bm{A}\bm{\hat{x}}_{k-1} \nonumber \\
    &+ \bm{P}_{k}\bm{C}^T\bm{R}^{-1}\bm{z}_k + \bm{P}_{k}\bm{E}^T\bm{Q}^{-1}\bm{\Delta }_k \nonumber \\
    &- \bm{P}_{k}\bm{E}^T(\bm{Q} + \bm{A}\bm{P}_{k-1}\bm{A}^T)^{-1}\bm{A}\bm{P}_{k-1}\bm{A}^T\bm{Q}^{-1}\bm{\Delta}_k \nonumber\\
    \overset{(2)}{=} &~\bm{P}_{k}\bm{E}^T(\bm{Q} + \bm{A}\bm{P}_{k-1}\bm{A}^T)^{-1}\bm{A}\bm{\hat{x}}_{k-1} \nonumber \\
    &+ \bm{P}_{k}\bm{C}^T\bm{R}^{-1}\bm{z}_k + \bm{P}_{k}\bm{E}^T \cdot \nonumber\\
    &\left\{\bm{Q}^{-1} - (\bm{Q} + \bm{A}\bm{P}_{k-1}\bm{A}^T)^{-1}\bm{A}\bm{P}_{k-1}\bm{A}^T\bm{Q}^{-1}\right\}\bm{\Delta}_k \nonumber \\
    \overset{(3)}{=} &~\bm{P}_{k}\bm{E}^T(\bm{Q} + \bm{A}\bm{P}_{k-1}\bm{A}^T)^{-1}\bm{A}\bm{\hat{x}}_{k-1} \nonumber \\
    &+ \bm{P}_{k}\bm{C}^T\bm{R}^{-1}\bm{z}_k + \bm{P}_{k|k}\bm{E}^T \Bigl\{(\bm{Q} + \bm{A}\bm{P}_{k-1}\bm{A}^T)^{-1} \cdot \nonumber \\
    &\left[(\bm{Q} + \bm{A}\bm{P}_{k-1}\bm{A}^T)\bm{Q}^{-1}- \bm{A}\bm{P}_{k-1}\bm{A}^T\bm{Q}^{-1}\right]\Bigl\}\bm{\Delta}_k \nonumber\\
      \overset{(4)}{=} &~\bm{P}_{k}\bm{E}^T(\bm{Q} + \bm{A}\bm{P}_{k-1}\bm{A}^T)^{-1}\bm{A}\bm{\hat{x}}_{k-1} \nonumber \\
    & +\bm{P}_{k}\bm{C}^T\bm{R}^{-1}\bm{z}_k + \bm{P}_{k}\bm{E}^T(\bm{Q} + \bm{A}\bm{P}_{k-1}\bm{A}^T)^{-1} \bm{\Delta}_k.  
\end{align}
The second equality follows from factoring out $\bm{P}_{k}\bm{E}^T$ and $\bm{\Delta}_k$ in the third and fourth summand, the third equality follows by factoring out $(\bm{Q} + \bm{A}\bm{P}_{k-1}\bm{A}^T)^{-1}$ in the curly brackets, and the fourth equality follows from the associativity of matrix multiplication. Hence, the final expression is given by the following recursion:
\begin{align}
    \bm{\hat{x}}_{k}= &~\bm{P}_{k}\bm{E}_k^T(\bm{Q} + \bm{A}_{k-1}\bm{P}_{k-1}\bm{A}_{k-1}^T)^{-1} \cdot \nonumber\\
     &~ (\bm{A}\bm{\hat{x}}_{k-1} + \bm{\Delta}_k)+\bm{P}_{k}\bm{C}^T\bm{R}^{-1}\bm{z}_k, \nonumber\\
     \bm{P}_{k}^{-1} = &~\bm{E}_k^T(\bm{Q} + \bm{A}_{k-1}\bm{P}_{k-1}\bm{A}_{k-1}^T)^{-1}\bm{E}_k + \bm{C}^T\bm{R}^{-1} \bm{C}.
     \label{covariance_ekf}
\end{align}
If $\bm{f}$ and $\bm{g}$ are linear, $\bm{\Delta}_k=\bm{0}$, and the recursion coincides with (\ref{state_kalman}). \par
Finally, we execute the algorithm iteratively in the form of the IEKF. Applying (\ref{covariance_ekf}) gives the new estimate iterate $\bm{\hat{x}}_{k}^{(i)}$. The Jacobians can be reevaluated for this estimate, and the procedure repeated until convergence; as a termination criterion we can take
$
    \left\| \bm{\hat{x}}_{k}^{(i)}- \bm{\hat{x}}_{k}^{(i-1)}\right \|_{\infty}\leq\epsilon,
$ for a small $\epsilon >0$. Algorithm \ref{alg1} provides the pseudo-code of the proposed IEKF for time step $k$.\par
\begin{algorithm}[!ht]
\caption{IEKF for under-determined DAE systems}
\label{alg1}
\textbf{Input:}  $\hat{\bm{x}}_{k-1|k-1}, \bm{P}_{k-1|k-1}$\\
\textbf{Output:}  $\hat{\bm{x}}_{k|k}, \bm{P}_{k|k}$
\begin{algorithmic}
\State{Calculate $\bm{A}_{k-1}$ from (\ref{eq:long})}
\State{$i \gets 0$}
\State{$\hat{\bm{x}}_{k|k}^{(0)} \gets \hat{\bm{x}}_{k-1|k-1}$}
\Do
\State {$i \gets i+1$}
\State \text{Calculate $\bm{E}_{k}^{(i)}, \bm{\Delta}_k^{(i)}$ from (\ref{eq:long})} 
\State {Apply (\ref{covariance_ekf}) to calculate $\hat{\bm{x}}_{k|k}^{(i)}$ and $\hat{\bm{P}}_{k|k}^{(i)}$}
\doWhile{$ \left\|\bm{\hat{x}}_{k|k}^{(i)}-\bm{\hat{x}}_{k|k}^{(i-1)}\right \|_{\infty}>\epsilon$}
\State{$\hat{\bm{x}}_{k|k} \gets \hat{\bm{x}}_{k|k}^{(i)}$}
\State{$\hat{\bm{P}}_{k|k} \gets \hat{\bm{P}}_{k|k}^{(i)}$}
\end{algorithmic}
\end{algorithm}

\section{Estimability analysis}
\label{Estimability}
The developments of the previous section rely on the assumption of full column rank of matrix $\left[\begin{smallmatrix}\bm{E}_k\\\bm{C} \end{smallmatrix} \right]$, also called the estimability condition \cite{kalman_descriptor}. Estimabilty refers to the ability to uniquely determine all state estimates given the dynamics and the measurement equations, or in other words: it refers to the existence of the unique minimizer of (\ref{opt_ekf}). If not satisfied, the Kalman filter information matrix, $\bm{P}_k^{-1}$ in (\ref{covariance_ekf}), is singular. It follows from \cite{kalman_descriptor} that the full column rank of $\left[\begin{smallmatrix}\bm{E}_k\\\bm{C} \end{smallmatrix} \right]$ is sufficient and necessary for estimability. This condition is also known as fast subsystem observability in the linear descriptor systems literature \cite{dai2014singular}. Note that both estimability and detectability assumptions are necessary for the convergence of the Kalman filter \cite{kalman_descriptor}. In this work, we limit our attention to estimability and leave detectability for future work.
% Also note that if all component models are available, ${\bm{E}}_4$ is a square matrix, and the full rank criterion follows from the DAE model being (locally) of differentiation index-1 \cite{power_dae}.
\subsection{Graphs of Structured Matrices}
The estimability criterion ensures, statistically speaking, the availability of sufficient amount of information to uniquely estimate all differential $\bm{y}_k$ and algebraic states $\bm{v}_k$ at all times. A direct computation of the rank of the given matrix is not only computationally expensive for large power systems, but more importantly, it does not provide any intuition about why the given sensor placement does not guarantee estimability nor where additional sensors should be placed to restore it. Instead, inspired by \cite{sse_topology} and \cite{scholtz}, where a topological criterion for static observability and a dynamic observer based on graph theory are presented, we aim to develop an algorithm that can be applied to any network and does not depend on the specific network parameters but only on the topological structure of the underlying grid. To this end, we introduce the concept of structured matrices.\par 
A structured matrix is a matrix whose elements are either fixed zeros or non-zero indeterminate (free) parameters.  The generic rank of a structured matrix is defined as the maximal rank it can reach as a function of its free parameters. In \cite{EC_full}, it was shown that the rank of a given matrix equals its generic rank except for pathological cases with Lebesgue measure zero. Such cases are unlikely to occur in practice, and even if they did, an arbitrary small perturbation of the parameters restores the maximal rank \cite{EC_full}. In SSE for power systems, such pathological cases are referred to as parametric unobservability \cite{sse_topology}. However, if the measurement Jacobian is generically full rank, the system is still called topologically observable. Analogously, we refer to the full generic rank of $\left[\begin{smallmatrix} \bm{E}_k \\ \bm{C} \end{smallmatrix}\right]$ as the \textit{topological estimability} criterion for the proposed DSE. \par 
We introduce a structured system as a pair of structured matrices $(\mathcal{E}, \mathcal{C})$, where $\mathcal{E} \in \mathbb{R}^{n \times n}$ and $\mathcal{C} \in \mathbb{R}^{m \times n}$. To each structured system, a directed graph can be associated in the following way: the graph of the given pair $(\mathcal{E},\mathcal{C})$ contains \mbox{$n+m$} vertices, denoted $v_1, v_2, \ldots, v_{n+m}$. For each free parameter $a_{ij}$ where $i \in 1, \ldots, n\!+\!m$ and $j \in 1, \ldots, n$ of the structured system, there exists an oriented edge from $v_j$ to $v_i$ in the graph representation. To facilitate the understanding, let us introduce an example of a structured pair.\par
\textit{Example}: Let
\begin{equation}
    \mathcal{E} = \begin{bmatrix}
    0 & a_{12} & 0\\
    0 & a_{22} & a_{23}\\
    a_{31} & 0& 0 
    \end{bmatrix}, \quad
    \mathcal{C} = \begin{bmatrix}
    0&0&a_{43}
    \end{bmatrix}.
\end{equation}
The associated graph is given in Figure \ref{fig:toygraph}.
\begin{figure}
    \centering
    \includegraphics[scale=0.6]{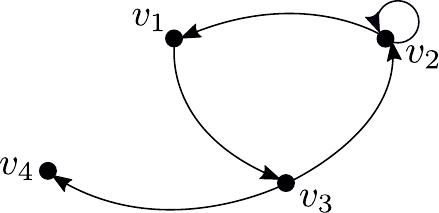}
    \caption{Graph of the structured pair $(\mathcal{E}, \mathcal{C})$}
    \label{fig:toygraph}
\end{figure}
Notice that the vertices corresponding to the rows of $\mathcal{C}$ are all sink vertices as no edge points outwards from them; we call all paths ending in such vertices $\mathcal{C}$-topped paths.\par
We now state a theorem from \cite{EC_full} useful for analyzing the topological estimability of power systems.
\begin{theorem}
\cite[Theorem 1]{EC_full} Given structured matrices $\mathcal{E} \in \mathbb{R}^{n \times n}$ and $\mathcal{C} \in \mathbb{R}^{m \times n}$, the matrix $\begin{bmatrix} \mathcal{E} \\ \mathcal{C}\end{bmatrix}$ is generically full column rank if and only if there exists a disjoint union of loops and $\mathcal{C}$-topped paths that span all vertices of $\mathcal{E}$ in the graph of the structured pair.
\label{th:graph}
\end{theorem}
Applying the previous theorem to the example in Figure \ref{fig:toygraph}, we ascertain that $\begin{bmatrix}\mathcal{E} \\ \mathcal{C} \end{bmatrix}$ is generically full column rank because its graph can be spanned by the simple path \mbox{$v_2\rightarrow v_1\rightarrow v_3\rightarrow v_4$}. However, removing, for example, edge $a_{31}$ renders the first column of $\begin{bmatrix} {\mathcal{E}} \\ {\mathcal{C}}\end{bmatrix}$ a zero column and the graph $\mathcal{E}$ cannot be spanned anymore by disjoint loops and $\mathcal{C}$-topped paths. \par
\subsection{Topological Estimability}
In what follows, we drop the subscript $k$ from the matrices in (\ref{eq:long}) to simplify the notation. 
First, we show that it is sufficient to certify the full rank of $\left[\begin{smallmatrix}\bm{E}_4\\\bm{C}_2 \end{smallmatrix} \right]$ to satisfy the estimability condition.
\begin{lemma}
For sufficiently small discretization steps $h$, full column rank of matrix $\begin{bmatrix}\bm{E}_{4}\\ \bm{C}_2 \end{bmatrix}$ implies full column rank of the matrix $\begin{bmatrix} \bm{E} \\ \bm{C}\end{bmatrix}$.
\end{lemma}
\begin{proof}
See Appendix~\ref{proof1}.
\end{proof}
Lemma 2 suggests that the placement of PMUs should ensure full rank of $\left[\begin{smallmatrix} \bm{E}_4 \\ \bm{C}_2 \end{smallmatrix}\right]$. This requirement is addressed in the next subsection. 
We now regard the pair $(\bm{E}_4, \bm{C}_2)$ as structured matrices, by considering their non-zero entries (which depend on the power system parameters and operating point) as free parameters. As in the previous subsection, we can associate a graph with the pair $(\bm{E}_4, \bm{C}_2)$, with $2(n_a + n_m)$ vertices and a structure similar to the power system topology. Henceforth, to make the distinction clearer, when we refer to the graph of the structured pair $(\bm{E}_4, \bm{C}_2)$, we use the terminology of vertices and edges, whereas for the topology of the corresponding power system, we use nodes and branches.  \par
\begin{comment}
To state the main theorem of the subsection, we need to define a representation of the phasor measurements in the power system topology. We assign each voltage phasor measurement to the node at which it is placed; the branch current flow measurements are assigned to either of the two nodes defining the branch, and the current injection measurements are assigned to the respective node or to any of its neighboring nodes. This selection is a consequence of the respective measurement functions $\bm{C}_2$ \cite{linear_PMU} and the fact that there can only be one path ending at each $\bm{C}_2$ vertex to satisfy the conditions of Theorem \ref{th:graph}.  
\end{comment}
\begin{theorem}
The power system (\ref{eq:dae_discrete_diff})--(\ref{eq:dae_discrete_meas}) is topologically estimable by (\ref{covariance_ekf}) if there exist disjoint paths along transmission lines from all nodes with an unknown injector to a node with an assigned PMU device. 
\end{theorem}
\begin{proof}
See Appendix~\ref{proof2}.
\end{proof}
Theorem 3 states that the number of available PMUs has to be at least as high as the number of unknown injectors. Hence, the more dynamic models are available, the fewer measurements are necessary for the uniqueness of the solution. The conclusion underscores how the static network observability requirements is relaxed. It is well known that SSE requires at least $2n-1$ ($n$ being the number of nodes) measurements to achieve the unique solution \cite{abur}.\par 
The above analysis provides a method to analyze estimability of the DSE. However, it does not suggest how to optimally place PMUs to achieve a specific performance metric. This can be achieved by formulating a combinatorial optimization problem that may be amenable to submodular optimization techniques -- a topic for future work.

\section{Numerical Results}
\label{Results}
\subsection{Simulation Setup}
To validate the performance of the proposed method, we perform numerical simulations on a power system comprising synchronous generators (SG), photovoltaic (PV) generators, loads, and a network interconnecting them. Other generator and load types can be integrated into the framework in a straightforward manner. \par 
The IEEE 39-bus test system is used, whose parameters are taken from \cite{hiskens}. The benchmark model is extended to include dynamic PV generator models from \cite{DER_A}, connected to nodes $20$, $21$, and $23$, each producing 100 MW at steady state. The ground-truth values are obtained by a simulation that employs the sixth order, subtransient synchronous generator model \cite{machkowski}, including the IEEE DC1A automatic voltage regulator \cite{machkowski} and the simple TGOV1 turbine model \cite{governor}. The subtransient reactances are set to 90\% of their transient counterparts, whereas the subtransient time constants are 100 times smaller than the transient counterparts \cite{kundur}. All loads are modeled as a composition of a negative power injection and an impedance. The simulation model is determined and comprises $110$ differential and $78$ algebraic states and, in total, $188$ equations. The solution is obtained by the collocation discretization method with third-degree Radau collocation points. \par
\subsection{Estimation Setup}
The goal is to perform the state estimation on the highlighted area depicted in Figure~\ref{fig:39bus}. From the perspective of the operator of this area, the rest of the power system is out of interest. Hence, interface node 16 can be considered as a node with an unknown dynamic model by the proposed estimation scheme, and only the topology and network parameters of the highlighted area are needed for the DSE. The area of interest can be chosen arbitrarily based on the responsibility of the system operator, access to the model and real-time measurements, and computational limitations. For the estimation, a simplified, fourth-order Anderson-Fouad's transient model of the synchronous generator is used \cite{milano2010power} to mimic the fact that mathematical models represent an approximate description of the true behavior of a physical system. We assume that the exciter, the turbine model, and the generators' power and voltage set points are known. However, we consider all PV and load models unknown to the estimator. Obtaining accurate models of many distributed PVs is challenging in practice; the same holds for time-varying residential and commercial loads in distribution grids, and modeling them as static consumers with their power inferred from historical data may be unreliable, especially during disturbances and transient events \cite{milano2010power}. If dynamic models or reliable real-time consumption data are available, they can be incorporated into the framework in a straightforward manner. \par
Therefore, all shaded elements in Figure \ref{fig:39bus}, connected to nodes $16, 20, 21, 23$, and $24$, are unknown to the estimator, which makes the estimation model incomplete. We assume no prior knowledge and make no assumptions regarding the injections of these nodes. The estimation model comprises $36$ differential states, $22$ algebraic states, and only $48$ equations. PMU measurements are assumed to arrive simultaneously every \SI{20}{ms} \cite{PMU_stand}; their values are obtained by adding Gaussian measurement noise with a standard deviation in the range 0.001--0.003 p.u. to the real and imaginary voltage and current ground-truth values, which is according to the IEEE standard \cite{PMU_stand}.\par
All simulations are implemented in Python on an Intel i7-10510U CPU @ 1.80GHz with 16GB RAM. Cholesky factorization is used to factorize the state estimation covariance matrix in (\ref{covariance_ekf}), which can be ill-conditioned. All Jacobians are computed by the algorithmic differentiation using CasADi's symbolic framework \cite{Andersson2019}.  
\begin{figure}
    \centering
    \includegraphics[width=\columnwidth]{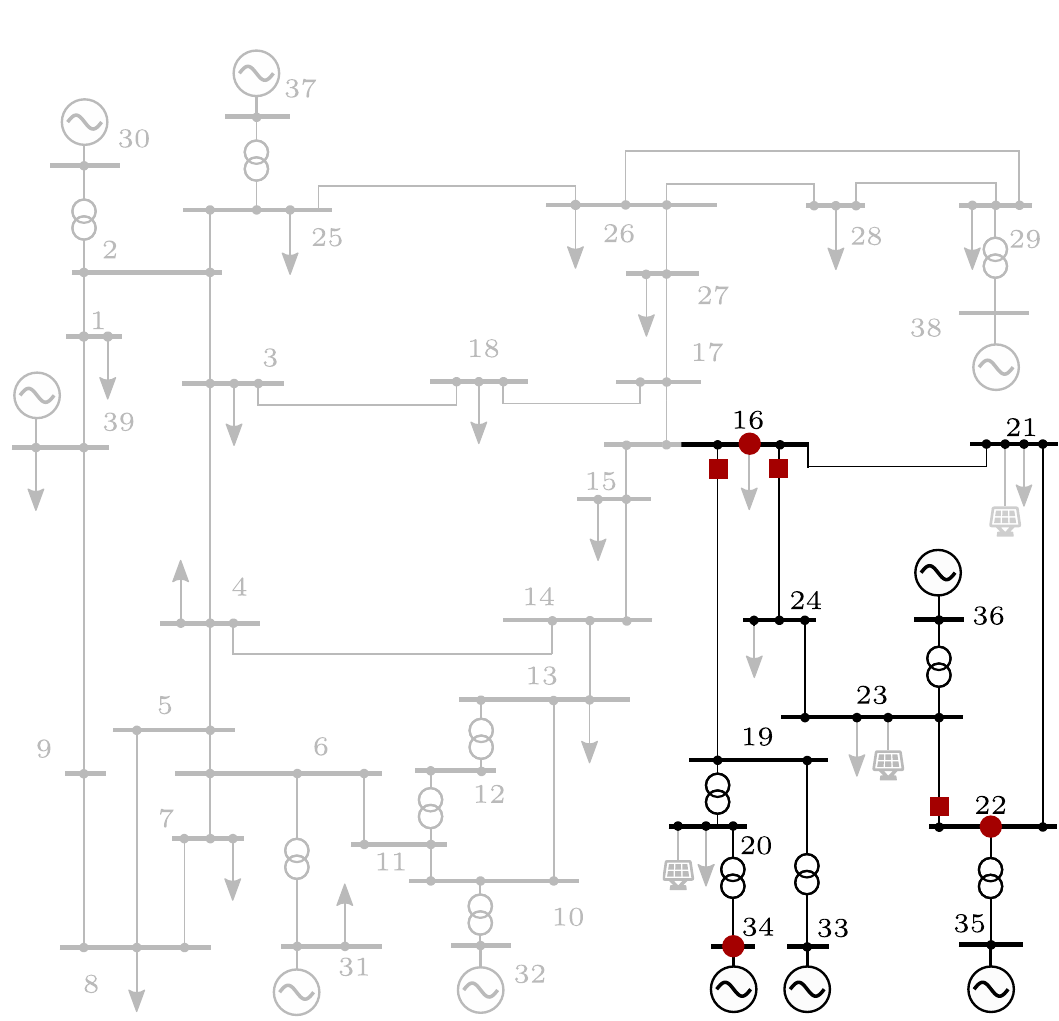}
    \caption{Single line diagram of the IEEE 39-bus test system. The estimation is performed on the highlighted area; shaded elements are unknown to the estimator. Voltage phasor measurements are denoted by the red circles, and current phasor measurements are denoted by the red squares.}
    \label{fig:39bus}
\end{figure}
\subsection{Case Study 1: Network Short Circuit}
Short circuits are one of the most severe disturbances with respect to their impact on power system stability, but also with respect to the stability of the implemented numerical methods for simulation. Here, we simulate a three-phase short circuit at $t=\SI{8}{s}$ in the middle between two PMU sampling events on the line between nodes $5$ and $8$, which is cleared after \SI{60}{ms} on both ends. Note that the operator of the shaded area from Figure~\ref{fig:39bus} (and hence the proposed estimator) has no direct awareness of this disturbance. The estimation is initialized with a voltage flat start; the differential states are initialized with a uniform random error from the ground truth with bounds $\pm 10\%$, except for the rotor speeds, where the bound is $\pm 0.1\%$. PMUs are placed to satisfy the conditions of Theorem 3. PMU voltage measurements are placed at nodes $19$, $23$, and $34$; PMU branch current measurements are located at lines $16-19$, $16-24$, and $22-23$. \par
To certify topological estimability, we assign the PMU at line $16-24$ to node $24$, and the PMU at line $22-23$ to node $22$ and observe that topological estimability (Theorem 3) for the proposed method is satisfied as all unknown injectors can be connected to a node with an assigned PMU device by disjoint paths: injector $16$ to node $16$, injector $20$ to node $34$, injector $21$ to node $22$, injector $23$ to node $23$, and finally injector $24$ to node $24$ (see Figure~\ref{fig:39bus}). By contrast, the power system is statically not observable, or, more precisely, ignoring the dynamic equations, the voltages are not uniquely estimable from the given PMU measurements. \par
\subsubsection{Standard Measurement Noise Magnitude}
The results for the rotor angle, rotor speed, turbine power, and exciter voltage for the generators at nodes $33$, $34$, and $36$ for the standard deviation of measurement noise of 0.001 p.u. are shown in Figure~\ref{fig:results_short}. All rotor angles are given with respect to the rotor angle at bus $36$. 

\begin{figure*}
        \includegraphics[]{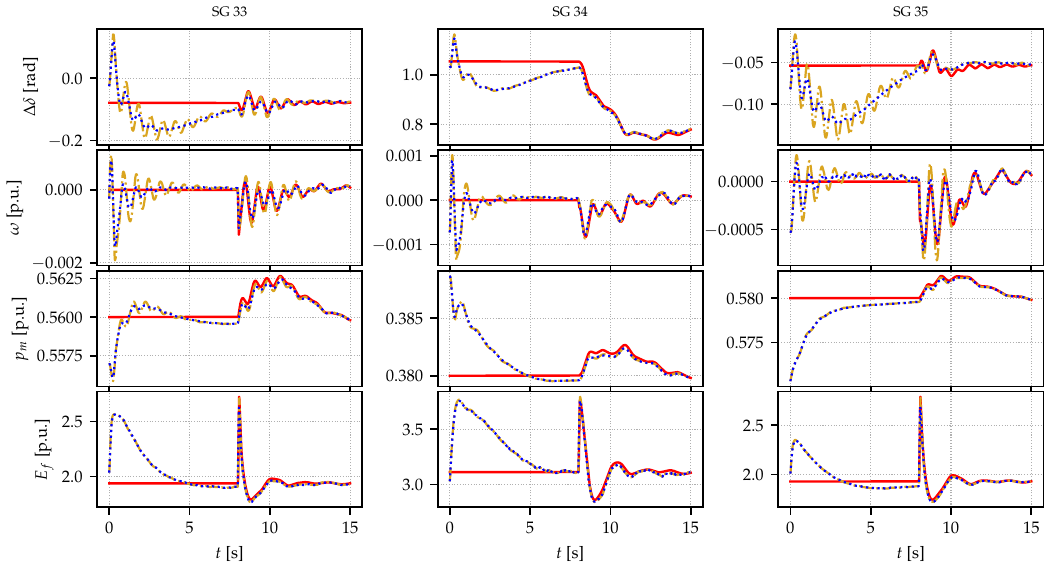}
        \caption{Estimation of dynamic states of SGs during initialization and during and after the short-circuit in the network. The true state is denoted by \protect\redline \hspace{0.03cm}; the proposed IEKF with the trapezoidal rule by \protect\goldline\hspace{0.03cm}, and the backward Euler by \protect\blueline\hspace{0.03cm}. From top to bottom: rotor angle, rotor speed, turbine power, exciter voltage; from left to right: SG 33, SG 34, and SG 35.}
        \label{fig:results_short}
\end{figure*}
It can be seen that the proposed method showcases excellent tracking properties both during initialization and during disturbances. Both the implicit Euler and the trapezoidal rule achieve satisfactory performance. The implicit Euler appears to have a more damped response to wrong initial conditions and better tracking during transients. The explicit Euler is not shown in the figure because it suffers, as expected for stiff systems, from oscillatory behavior.\par
In what follows, we only consider the backward Euler discretization scheme, as its tracking performance during transients is deemed superior to the forward Euler and the trapezoidal scheme. Figure~\ref{fig:results_short_voltage} presents the true and estimated voltage magnitudes of all nodes in the area of interest. The estimated states closely track the true values, both during steady state and transients. The average mean square error of the estimation of all nodes after the initialization phase ($t \geq \SI{7.5}{s}$) is $8.25 \cdot 10^{-7}$. This value is smaller than the variance of the PMU measurement noise despite the fact that the network is statically unobservable.\par
The maximal number of iterations to reach the prespecified threshold of $\epsilon = 10^{-4}$ and the average and worst-case computing times for the iterated schemes are shown in Table~\ref{tab:computational}. For both discretization schemes, the average and maximal computational times are below the typical PMU sampling rate.
\begin{table}[]
    \caption{Computational performance of the proposed DSE for two discretization methods at each PMU scan.}
    \centering
    \begin{tabular}{c|c c c c c|}
    \toprule
    Discr. Scheme& Max. Iter.&Aver. Iter.&Max. Time&Aver. Time \\
    \midrule
    Backward Euler &  3 & 2.13&\SI{5.49}{ms}&\SI{2.38}{ms} \\
    Trapezoidal  &3 & 2.10&\SI{5.98}{ms}& \SI{2.49}{ms}\\
    
    \bottomrule
    \end{tabular}
    \label{tab:computational}
\end{table}
\begin{figure}
    \includegraphics[]{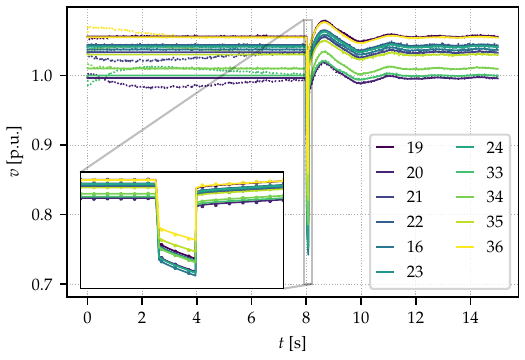}
    \caption{Estimated and true voltage magnitudes during the initialization and during and after the short circuit. Full lines denote the true voltages; markers denote the estimated ones.}
    \label{fig:results_short_voltage}
\end{figure}
\par 
\subsubsection{Elevated Measurement Noise Magnitude}
In accordance with findings from \cite{Wang2018}, the actual PMU measurement noise may deviate from Gaussian distribution and exhibit thick tales. To evaluate the robustness of the proposed method against such scenarios, we conduct additional simulations in which PMU data are contaminated with Laplacian noise featuring a standard deviation of 0.003 p.u. This value is three times larger than the standard deviation in the previous experiments. Due to space limitations, only states of synchronous generators 33 and 36 are shown in Figure~\ref{fig:results_short_laplace}.
\begin{figure}
    \includegraphics[]{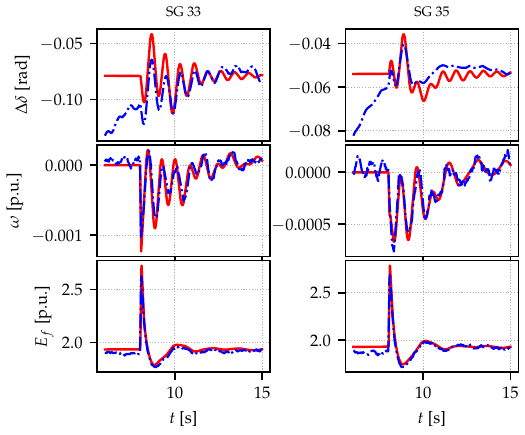}
    \caption{Results of estimated and true voltage magnitudes during and after the short circuit for elevated PMU noise levels following Laplace distribution. Red lines denote the ground truth values; blue lines denote the estimated ones.}
    \label{fig:results_short_laplace}
\end{figure}
Notably, the figure reveals that, even in the presence of significantly elevated noise levels, the estimator demonstrates resilience and remains capable of accurately tracking the true states of the system. \par
\subsubsection{Impact of PMU placement}
Next, we perform several different simulations for different PMU placements, all compliant with the condition stipulated in Theorem 3. Table \ref{tab:pmu} provides a comprehensive overview of the performance of different measurement constellations. The mean square error for selected differential and algebraic states remains small across all tested configurations.
\begin{table*}[t]
    \caption{Impact of the PMU placement on the accuracy of DSE.}
    \centering
    \begin{tabular}{c c| c c c c c}
    \toprule
    \multicolumn{2}{c|}{PMU Measurements} &\multicolumn{5}{c}{Mean Square Error} \\
    Voltage&Current &$\delta_{\rm{}}$&$\omega_{\rm{}}$&$E_{f{\rm{}}}$&$p_{m{\rm{}}}$&$v_{\rm{}}$ \\
    \midrule
     19, 23, 34 &  16-19, 16-24, 22-23 & $6.27 \cdot 10^{-5}$&$3.72 \cdot 10^{-9}$&$5.18 \cdot 10^{-4}$&$4.62 \cdot 10^{-8}$&$9.32 \cdot 10^{-7}$\\
    19, 23, 24 &  16-19, 21-22, 22-23 & $7.34 \cdot 10^{-5}$&$3.45 \cdot 10^{-9}$&$7.52 \cdot 10^{-4}$&$4.47 \cdot 10^{-8}$&$1.46 \cdot 10^{-6}$\\
    20, 35 &  16-19, 16-24, 23-24, 35-22 &$1.46 \cdot 10^{-5}$&$4.04 \cdot 10^{-9}$&$4.94 \cdot 10^{-4}$&$1.53 \cdot 10^{-8}$&$1.90 \cdot 10^{-6}$\\
    20, 21, 35 &  34-20, 16-24, 23-24, 35-22 & $9.26 \cdot 10^{-6}$&$4.00 \cdot 10^{-9}$&$8.04 \cdot 10^{-4}$&$1.28 \cdot 10^{-8}$&$2.40 \cdot 10^{-6}$\\
    20, 21, 24, 33 &  34-20, 16-24, 21-22 & $5.79 \cdot 10^{-6}$&$2.71 \cdot 10^{-9}$&$7.58 \cdot 10^{-4}$&$1.76 \cdot 10^{-8}$&$1.14 \cdot 10^{-6}$\\
    \bottomrule
    \end{tabular}
    \label{tab:pmu}
\end{table*}
While the achieved results showcase the efficacy of SE with various PMU placements meeting the criteria of Theorem 3, it is essential to note that the selection of the locations of PMUs is not unique. Multiple combinations satisfy the conditions, resulting in different accuracy of state estimation, as detailed in Table \ref{tab:pmu}. In light of these findings, our future research will focus on formulating an optimization problem aimed at guiding the optimal placement of PMUs.\par

\subsection{Case Study 2: Load Disturbance}
To compare the proposed method with the existing methods in the literature, we simulate a step change of the load consumption at node $21$, mimicking a sudden disconnection of a large consumer. For comparison, we employ a decoupled approach from the literature, wherein, in the first stage, a static estimator calculates the voltages in the grid.  As the network is not observable with the given PMU configuration, additional power pseudo-measurements are used for all the loads and for the interface node 16. These values are kept unknown to the proposed IEKF. They can be obtained, for example, from historical data or from slow scan rate conventional measurements and correspond to the actual operating point before the disturbance. With more uncertainty regarding those values compared to the PMU measurements, the static state estimator is implemented with a robust least absolute value procedure (LAV). In the second stage, to estimate the differential states of the generators, the EKF with and the $H_{\infty}$ EKF (HEKF) \cite{Hinf_compare} are implemented. The terminal voltages are used as inputs, and the terminal currents as measurements. 

The results in Figure~\ref{fig:compare} highlight that all three methods accurately track the system states before the disturbance at $t=\SI{8}{s}$. However, after the disturbance, the robust comparison methods exhibit a steady state error due to unexpected load change. The proposed estimator achieves accurate state tracking before and after the disturbance.

These results confirm that modeling loads and renewables with their forecasted infeed can lead to biased estimates during sudden disturbances, even if robust estimation methods are used. \par 
\begin{figure}
    \includegraphics[]{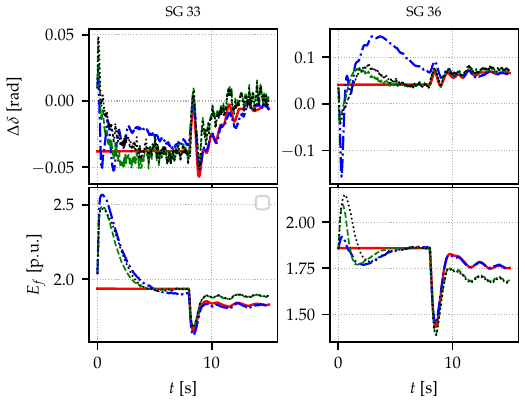}
    \caption{Comparison of the proposed IEKF \protect\blueline\hspace{0.03cm} with the LAV+EKF \protect\blackline\hspace{0.03cm} and LAV+HEKF \protect\greenline\hspace{0.03cm}. The true state is represented by \protect\redline.  }
    \label{fig:compare}
\end{figure}

\section{Conclusions and future work}
\label{Conclusion}
A novel centralized DSE has been proposed to address the challenges arising from partially known power system models with no prior knowledge or assumptions regarding some of the injectors. In contrast to the standard IEKF, which only deals with linear ODE models, the proposed method treats incomplete nonlinear differential-algebraic-equation models. The estimation of all system states was demonstrated to be very accurate when the estimability condition was met. The key insight lies in the requirement that all unknown injectors be connectable to a unique PMU device through mutually disjoint paths. The numerical results showcase the importance of the employed discretization scheme for DSE with an incomplete nonlinear DAE model, and the moderate computing times of the proposed scheme signify real-time feasibility. The proposed method outperforms the robust two-stage decoupled methods from the literature, particularly during sudden disturbances. \par
Future work will analyze the detectability/observability of the employed DAE power system models and subsequently design an optimization algorithm for the optimal placement of PMUs. Next, even though the method can disregard injectors with unknown models, bad PMU data are still a major concern. Therefore, future work will also address bad data detection and identification.

\appendix 
\subsection{Kalman Filtering for Linear Discrete-Time Descriptor Systems}
\label{subsec_kf}
For improved readability, this subsection summarizes the most related results from \cite{ishihara}, where a recursive optimal filter for a general class of linear discrete-time, stochastic descriptor systems of the form
\begin{align}
    \label{eq:daed}
    \bm{\bar{E}} \bm{\bar{x}}_{k} &= \bm{\bar{A}}\bm{\bar{x}}_{k-1} + \bm{\bar{\xi}}_k,\\
    \bm{\bar{z}}_{k} &= \bm{\bar{C}} \bm{\bar{x}}_{k} + \bm{\bar{\nu}}_k,
    \label{eq:daea}
\end{align}
is presented. Matrices $\bm{\bar{E}}$ and $\bm{\bar{A}}$ need not be invertible or even square: the method can handle under/over-constrained systems. It is assumed that $\bm{\bar{\xi}}_k$ and $\bm{\bar{\nu}}_k$ are zero mean white Gaussian noise sequences with
\begin{equation}
    \mathbb{{E}}\left[
        \begin{bmatrix}
            \bm{\bar{\xi}}_k\\ \bm{\bar{\nu}}_k
        \end{bmatrix}\begin{bmatrix}
            \bm{\bar{\xi}}_j\\ \bm{\bar{\nu}}_j
        \end{bmatrix}^T \right]=\begin{bmatrix}
            \bm{\bar{Q}} & \bm{0}\\\bm{0} & \bm{\bar{R}}
        \end{bmatrix}\delta(k-j),
\end{equation}
where $\delta(l) = 1$ if $l=0$ and $0$ otherwise. $\bm{\bar{Q}}$ and $\bm{\bar{R}}$ are positive definite covariance matrices. For the more general case with positive semi-definite covariance matrices, the interested reader is referred to \cite{kalman_descriptor}. The key idea is to process both the time-evolution equation (\ref{eq:daed}) and the measurement equation (\ref{eq:daea}) simultaneously in a batch-mode regression. \par
The optimal recursive state estimation at time step $k$ can be formulated as the following optimization problem:
\begin{equation}
    \min_{\mathfrak{\bar{X}}} \quad (\mathfrak{\bar{A}}\mathfrak{\bar{X}}-\mathfrak{\bar{B}})^T\mathfrak{\bar{R}}(\mathfrak{\bar{A}}\mathfrak{\bar{X}}-\mathfrak{\bar{B}}) \label{opt_pdf},
\end{equation}
where\\
$\mathfrak{\bar{A}}:=\begin{bmatrix}\bm{\bar{E}} & \bm{\bar{A}}\\ \bm{\bar{C}} &\bm{0} \\ \bm{0}& \bm{I} \end{bmatrix}$, $\mathfrak{\bar{B}}:=\begin{bmatrix}\bm{0}\\ \bm{\bar{z}}_k \\ - \bm{\hat{\bar{x}}}_{k-1|k-1} \end{bmatrix}$, \\
$\mathfrak{\bar{R}} := \begin{bmatrix} \bm{\bar{Q}}^{-1} & \bm{0} & \bm{0} \\ \bm{0} & \bm{\bar{R}}^{-1} & \bm{0} \\ \bm{0} & \bm{0} & \bm{\bar{P}}^{-1}_{k-1|k-1} \end{bmatrix}$, $\mathfrak{\bar{X}} : = \begin{bmatrix} \bm{\bar{x}}_k \\ -\bm{\bar{x}}_{k-1}\end{bmatrix}.$\\
If the matrix $\begin{bmatrix} \bm{\bar{E}} \\ \bm{\bar{C}}\end{bmatrix}$ has full column rank, the optimal solution is given by \cite{kalman_descriptor}:
\begin{equation}
\hat{\mathfrak{\bar{X}}}=(\mathfrak{\bar{A}}^T\mathfrak{\bar{R}}\mathfrak{\bar{A}})^{-1}\mathfrak{\bar{A}}^T\mathfrak{\bar{R}} \mathfrak{\bar{B}},
\end{equation}
or
\begin{align}
    \begin{bmatrix} \bm{\hat{\bar{x}}}_{k|k} \\ \bm{\hat{\bar{x}}}_{k-1|k} \end{bmatrix}=  \begin{bmatrix} \bm{\bar{P}}_{k|k} & \bm{\bar{P}}_{k,k-1|k}\\ \bm{\bar{P}}^T_{k,k-1|k} & \bm{\bar{P}}_{k-1|k} \end{bmatrix}
    \begin{bmatrix}\bm{\bar{C}}^T \bm{\bar{R}}^{-1}\bm{\bar{z}}_k \\ -\bm{\bar{P}}_{k-1|k-1}^{-1} \bm{\hat{\bar{x}}}_{k-1|k-1}\end{bmatrix},
    \label{kf_big_matrix}
\end{align}
where
\begin{equation}
\begin{aligned}
    &\begin{bmatrix} \bm{\bar{P}}_{k|k} & \bm{\bar{P}}_{k,k-1|k}\\ \bm{\bar{P}}_{k,k-1|k}^T & \bm{\bar{P}}_{k-1|k} \end{bmatrix} := \\
    &\begin{bmatrix} \bm{\bar{E}}^T\bm{\bar{Q}}^{-1}\bm{\bar{E}} + \bm{\bar{C}}^T\bm{\bar{R}}^{-1}\bm{\bar{C}} & \bm{\bar{E}}^T\bm{\bar{Q}}^{-1} \bm{\bar{A}}\\ \bm{\bar{A}}^T\bm{\bar{Q}}^{-1}\bm{\bar{E}} & \bm{\bar{A}}^T \bm{\bar{Q}}^{-1}\bm{\bar{A}} + \bm{\bar{P}}_{k-1|k-1}^{-1} \end{bmatrix}^{-1}. \label{kf_inverse}
        \end{aligned}
\end{equation}
Notice that only $\bm{\bar{P}}_{k|k}$ and $\bm{\bar{P}}_{k,k-1|k}$ are needed to calculate the most recent state estimate. Applying the matrix inversion lemma and some algebraic transformations to (\ref{kf_inverse}) gives
\begin{align}
    \bm{\bar{P}}_{k|k}^{-1} &= \bm{\bar{E}}^T(\bm{\bar{Q}} + \bm{\bar{A}}\bm{\bar{P}}_{k-1|k-1}\bm{A}^T)^{-1}\bm{\bar{E}} + \bm{\bar{C}}^T\bm{\bar{R}}^{-1} \bm{\bar{C}}, \label{cov_kalman} \nonumber\\
    \bm{\bar{P}}_{k,k-1|k} &= - \bm{\bar{P}}_{k|k}\bm{\bar{E}}^T(\bm{\bar{Q}} + \bm{\bar{A}}\bm{\bar{P}}_{k-1|k-1}\bm{\bar{A}}^T)^{-1}\bm{\bar{A}}\bm{\bar{P}}_{k-1|k-1}.
\end{align}
For details, the interested reader is referred to \cite{mhe_kalman}. Substituting these expressions into (\ref{kf_big_matrix}) yields
\begin{align}
    \hat{\bm{\bar{x}}}_{k|k}&= \bm{\bar{P}}_{k|k}\bm{\bar{E}}^T(\bm{\bar{Q}} + \bm{\bar{A}}\bm{\bar{P}}_{k-1|k-1}\bm{\bar{A}}^T)^{-1}\bm{\bar{A}}\bm{\hat{\bar{x}}}_{k-1|k-1} \nonumber \\
    &~~~+ \bm{\bar{P}}_{k|k}\bm{\bar{C}}^T\bm{\bar{R}}^{-1}\bm{\bar{z}}_k. \label{state_kalman}
\end{align}
For standard state space systems $(\bm{E} = \bm{I})$, the recursion \mbox{(\ref{cov_kalman})--(\ref{state_kalman})} reduces to the classical Kalman filter in the one-step form (see \cite[p.~131]{simon2006optimal}).
\subsection{Proof of Lemma 2}
\label{proof1}
\begin{proof}
First, we introduce the following notation:
\begin{equation*}
\Tilde{\bm{E}}_3 := \begin{bmatrix} \bm{E}_3 \\ \bm{0} \end{bmatrix} \in \mathbb{R}^{({n}_g + n_m) \times n_d}, ~ \Tilde{\bm{E}}_4 := \begin{bmatrix} \bm{E}_4 \\ \bm{C}_2 \end{bmatrix}\in \mathbb{R}^{({n}_g + n_m)\times n_a},
\end{equation*}
to rewrite
$
       \begin{bmatrix} \bm{E} \\ \bm{C}\end{bmatrix} =  \begin{bmatrix}
    h\bm{E}_1 - \bm{I} & h\bm{E}_2 \\ \Tilde{\bm{E}}_3 & \Tilde{\bm{E}}_4
    \end{bmatrix}. 
$
For sufficiently small $h$, the submatrix $\begin{bmatrix}h \bm{E}_1 - \bm{I}\end{bmatrix}$ is full rank. \par 
To complete the proof, assume for the sake of contradiction that the matrix $\begin{bmatrix} \bm{E} \\ \bm{C} \end{bmatrix}$ is not full column rank. Then there exists a vector $\begin{bmatrix} \bm{x}_1 \\ \bm{x}_2 \end{bmatrix}\neq \bm{0}$ of appropriate dimension such that 
\begin{equation}
    \begin{bmatrix}
    h\bm{E}_1 - \bm{I} & h\bm{E}_2 \\ \Tilde{\bm{E}}_3 & \Tilde{\bm{E}}_4
    \end{bmatrix}
    \begin{bmatrix}
    \bm{x}_1 \\ \bm{x}_2
    \end{bmatrix} = 
    \begin{bmatrix}
    \bm{0} \\ \bm{0}
    \end{bmatrix}.
    \label{rank1}
\end{equation}
Expressing $\bm{x}_1$ from the first equation of (\ref{rank1}) as
\begin{equation}
    \bm{x}_1 = h (\bm{I} - h{\bm{E}}_1)^{-1}\bm{E}_2 \bm{x}_2,
    \label{eq:x_1=0}
\end{equation} 
and substituting into the second gives
\begin{equation}
    \left(h \Tilde{\bm{E}}_3 (\bm{I} - h \bm{E}_1)^{-1} \bm{E}_2 + \Tilde{\bm{E}}_4 \right) \bm{x}_2 = \bm{0}.
    \label{rank3}
\end{equation}
Premultiplying (\ref{rank3}) by $\frac{1}{h} \left(\Tilde{\bm{E}}_4^T\Tilde{\bm{E}}_4\right)^{-1}\Tilde{\bm{E}}_4^T$ (the inverse exists because of full column rank of $\Tilde{\bm{E}}_4$) gives
\begin{equation}
    \left((\Tilde{\bm{E}}_4^T\Tilde{\bm{E}}_4)^{-1}\Tilde{\bm{E}}_4^T \Tilde{\bm{E}}_3 (\bm{I} - h\bm{E}_1)^{-1}\bm{E}_2 + \frac{1}{h}\bm{I} \right)\bm{x}_2=\bm{0}.
    \label{rank4}
\end{equation}
Because the Jacobians are bounded, and, for small $h$, $\begin{bmatrix}\bm{I} - h \bm{E}_1\end{bmatrix}^{-1}$ is also bounded, as $h \to 0$, the eigenvalues of $\frac{1}{h}\bm{I}$ are dominating, and the matrix from (\ref{rank4}) is full rank. Consequently, $\bm{x}_2=\bm{0}$, and from (\ref{eq:x_1=0}), it follows that $\bm{x}_1=\bm{0}$, which is a contradiction and therefore concludes the proof. 
\end{proof}
\subsection{Proof of Theorem 3}
\label{proof2}
\begin{proof}
Before stating the proof, we need to analyze the structure of $\bm{E}_4$ and $\bm{C}_2$. Note that each node of the power system is represented by two vertices in the graph representation; we call them the real and the imaginary vertex of the corresponding power system node. By definition, $\bm{E}_4$ is the Jacobian of the algebraic current balance equations with respect to the complex voltages expressed in rectangular coordinates. Therefore, $\bm{E}_4$ has a similar structure to the power system admittance matrix, but the rows corresponding to the nodes with unknown injectors are replaced by zero rows, and the sensitivities of the dynamic models are added. Hence, there exists a free parameter in $\bm{E}_4$ if and only if the real/imaginary current balance equation in a particular node depends on the real/imaginary voltage of another node. It follows that edges exist in the graph of the structured system only between neighboring nodes of the power system. Furthermore, all edges in the corresponding graph are bidirectional, except for the edges originating from the vertices corresponding to the nodes with unknown injectors (the corresponding rows in matrix $\bm{E}_4$ are zero rows). Moreover, if two nodes $i$ and $j$ are connected by a transmission line with conductance $G_{ij} \neq 0$, there exist edges between their real vertices and between their imaginary vertices of the corresponding nodes; if the susceptance $B_{ij} \neq 0$, there exist edges between the real and the imaginary vertices of the corresponding nodes. Shunt elements correspond to either a loop around a vertex itself or an edge from its real/imaginary counterpart and also they ensure the independence of the entries in each row of $\bm{E}_4$. \par
Let us also define the representation of the phasor measurements in the graph of the structured pair. We assign each voltage phasor measurement to the node at which it is placed; the branch current flow measurements are assigned to either of the two nodes defining the branch, and the current injection measurements are assigned to the respective node or to any of its neighboring nodes. This selection is a consequence of the respective measurement functions $\bm{C}_2$ \cite{linear_PMU} and the fact that there can only be one path ending at each $\bm{C}_2$ vertex to satisfy the conditions of Theorem \ref{th:graph}. Note that the edges leading to measurement vertices are unidirectional (see \cite{EC_full}). \par
Now, by Theorem \ref{th:graph}, we only need to show that, under the given assumptions, $\bm{E}_4$ vertices in the graph representation of $(\bm{E}_4, \bm{C}_2)$ can be spanned by a disjoint union of loops and $\bm{C}_2$-topped paths. Indeed, starting from the origin vertices (corresponding to the nodes with unknown injectors), we move in the same way as in the power systems topology, considering that we need to traverse two edges for each transmission line. Thus, whenever we move along a transmission line from node $i$ to node $j$ with ${G}_{ij} \neq 0$, we move from the real vertex of node $i$ to the real vertex of $j$ and from the imaginary vertex of $i$ to the imaginary vertex of node $j$; if ${B}_{ij} \neq 0$, we move from the real to the imaginary and from the imaginary to the real vertex; if ${G}_{ij}{B}_{ij} \neq 0$, we can choose either of the two options. Once a node with an assigned measurement device is reached, we also reach a $\bm{C}_2$ sink vertex for both paths, as we assume that all PMUs measure both real and imaginary components. In this way, all paths stay disjoint. Recall that there are only outward-going edges from a vertex with an unknown injector and only inward-going edges towards a measurement vertex, and they are both traversed in the correct direction. Now, we only need to span the remaining vertices that have not been traversed so far. Those remaining vertices correspond to the nodes with the known injectors' dynamic models or zero injection nodes, hence to nonzero rows in matrix $\bm{E}_4$. As such, they can either be spanned by trivial loops around themselves or around their real/imaginary counterpart vertex. 
\end{proof}
We note that the previous analysis is valid irrespective of the employed dynamic models. To show that, observe that the sensitivities of the dynamic models' current injections are only a function of the voltage of that node. Hence, these values are added to a corresponding $2 \times 2$ submatrix, where the diagonal or off-diagonal entries -- due to at least one incoming transmission line -- enable forming a loop around both vertices.\par

\bibliography{main.bib}
%\
\bibliographystyle{ieeetr}

\newpage

\vfill

\end{document}